\documentclass[10pt,twocolumn, final, conference]{IEEEtran}
\usepackage{amssymb, amsmath, graphicx, cite, color, epsfig, subfigure}

\newcommand{\nr} { {\sf r} }

\newcommand{\none} { {\sf 1} }
\newcommand{\ntwo} { {\sf 2} }

\newtheorem{thm}{Theorem}[section]

\newtheorem{theorem}{Theorem}
\newtheorem{lemma}[theorem]{Lemma}

\allowdisplaybreaks[1]

\begin{document}

\title{List decoding for nested lattices and applications to relay channels}
\author{Yiwei Song and Natasha Devroye
\thanks{Yiwei Song and
Natasha Devroye are with the Department of Electrical and Computer
Engineering, University of Illinois at Chicago, Chicago, IL 60607.
Email: ~ 	ysong34, devroye@uic.edu.} }

\maketitle

\begin{abstract}
 We demonstrate a decoding scheme for nested lattice codes which is able to decode a list of a particular size which contains the transmitted codeword with high probability.  This list decoder is analogous  to that used in random coding arguments in achievability schemes of relay channels,  and allows for the effective combination of information from the relay and source node.
   Using this list decoding result, we demonstrate 1) that lattice codes may achieve the capacity of the physically degraded AWGN relay channel, 2) an achievable rate region for the two-way relay channel with direct links using lattice codes, and 3) that we may improve the constant gap to capacity for specific cases of the two-way relay channel with direct links.
\end{abstract}

\section{Introduction}

\subsection{Motivation.}

Lattice codes have shown great recent promise in multi-user additive white Gaussian noise (AWGN) channels. While they are well known to be both good source and channel codes for Gaussian sources and channels respectively for point-to-point systems \cite{Erez:2004}, they are able to achieve capacity in certain multi-terminal AWGN channels as well including AWGN broadcast channels \cite{zamir2002nested} and multiple access channels \cite{nazer2009compute}. In three user (with logical extensions to $>3$ users) interference channels, their structure has enabled the decoding of  (portions of) ``sums of interference'' terms \cite{bresler_tse:2008, jafar:very_strong_IC}, allowing receivers to ``subtract'' off this sum rather than try to decode individual interference terms in order to remove them.

Lattices have also been of significant use in two-way Gaussian relay channels. The two-way relay channel consists of three nodes: two terminal nodes $\none, \ntwo$ that wish to exchange their two independent messages through the help of one relay node $\nr$. In particular, when no direct link is present between the terminal nodes and all information must pass through the relay, having the terminal nodes employ {\it nested lattice codes}, which ensures that their sum is again a lattice point allows for the sum of the two terminal node signals to be decoded at the relay. Sending this sum (possibly re-encoded) allows the terminal nodes to exploit their own message side-information to recover the other user's message \cite{Narayanan:2008,nam:2009bit}. Gains over decode-and-forward schemes where both terminals transmit simultaneously to the relay (as in full-duplex operation or two-phase MABC protocols \cite{Kim:sarnoff} for half-duplex nodes) stem from the fact that, if using random Gaussian codebooks, the relay will see a multiple-access channel and require the decoding of both individual messages, even though only the sum is needed. In contrast, no multiple-access (or sum-rate) constraint is imposed by the lattice decoding of the sum - leading to rate gains under certain channel conditions.

Lattices have equally found their place in achievability schemes for the multi-way relay channel \cite{gunduz2010multi, sezgin2010divide}, where groups of users wish to exchange messages through a relay. In particular, lattice codes are combined with random codes, superposition and successive decoding schemes to exploit gains similar to those seen for two-way relay channels.

\subsection{Contributions.}

As can be seen from the above applications, lattice codes may in some cases be used ``almost'' like random codes; see \cite{zamir-lattices} for a comprehensive survey of where lattices are useful and sometimes outperform random codes.  Aside from fairly general results on the use of lattices in relay networks \cite{nazer2009compute, nam:2009nested, ozgur2010approximately},  the usage of lattice codes in scenarios where information flows from source to destination along two paths, as in the classical one-way relay channel \cite{Cover:1979} remains relatively unexplored. In the relay channel, one of the most fundamental decode-and-forward schemes is that of Thm.1 of \cite{Cover:1979}, in which Block Markov superposition and random binning are used at the encoder, and joint-typicality-based successive decoding  using  a list decoder is employed at the decoder in order to efficiently merge the information available along the two paths.
Our contributions are: 

\smallskip
\noindent $\bullet$ we define a lattice equivalent of a list decoder;

\smallskip
\noindent $\bullet$ we use this list decoder to show that lattice codes achieve the capacity of the physically degraded AWGN relay channel;

\smallskip
\noindent $\bullet$ we use this to provide smaller constant gap to capacity results for some types of  Gaussian  two-way relay channels with direct links  \cite{Kim:sarnoff, xie2007network, avestimehr2009capacity}. This is the first application of lattice codes to the two-way channel model with direct links. 

\subsection{Paper layout.}

We introduce nested lattice codes and our lattice list decoding scheme in Section \ref{sec:lattice}; in Section \ref{sec:degraded} we show how this may be applied to achieve the capacity of the degraded AWGN relay channel using nested lattice codes; in Section \ref{sec:two-way} we demonstrate improved finite-gap results for the two-way relay channel with direct links. 

\section{A list decoder for nested lattice codes}
\label{sec:lattice}

We first introduce previous work on lattice codes as well as our notation. We then propose an encoding and decoding scheme for ``list decoding'' of  lattice codes of rate $R$ over an AWGN channel of noise power $N$ and transmit power constraint $P$. Finally, we prove that this scheme can decode a list of size $2^{n(R-C(P/N))}$, for $C(x):=\frac{1}{2}\log_2\left(1+x\right)$, of  possible codewords which contains the correct one with high probability as the blocklength $n\rightarrow \infty$.

\subsection{Lattice notation}
\label{sec:notation}

We outline our notation for (nested) lattice codes for transmission over AWGN channels; comprehensive treatments may be found in \cite{loeliger1997averaging, zamir2002nested,  Erez:2004} and in particular \cite{zamir-lattices}; our presentation follows that of  \cite{zamir2002nested, nam:2009nested}.  An $n$-dimensional lattice $\Lambda$ is a discrete subgroup of Euclidean space $\mathbb{R}^n$ (of vectors ${\bf x}$) with Euclidean norm $|| \cdot ||$ under vector addition and may be expressed as all integral combinations of basis vectors ${\bf g_i}\in {\mathbb R}^n$
\[ \Lambda = \{ \lambda = G \; {\bf i}: \; {\bf i}\in \mathbb{Z}^n\},\]
for $\mathbb{Z}$ the set of integers, and $G := [{\bf g_1} | {\bf g_2}| \cdots {\bf g_n}]$ the $n\times n$ generator matrix corresponding to the lattice $\Lambda$. Further define:

$\bullet$ The {\it nearest neighbor lattice quantizer} of $\Lambda$ as \[ Q({\bf x}) = \arg \min_{\lambda\in \Lambda} ||{\bf x}-\lambda||;\]

$\bullet$ The { \texttt{mod }$\Lambda$} operation as ${\bf x}$  \texttt{mod }$\Lambda : = {\bf x} - Q({\bf x})$;

$\bullet$ The {\it fundamental region of $\Lambda$} as the set of all points closer to the origin than to any other lattice point \[\mathcal{V}:= \{{\bf x}:Q({\bf x}) = {\bf 0}\}\]
which is of volume $V: = \mbox{Vol}({\mathcal V})$.

$\bullet$ The {\it second moment per dimension of a uniform distribution over ${\mathcal V}$} as
\[ \sigma^2(\Lambda) : = \frac{1}{V}\cdot \frac{1}{n} \int_{\mathcal V} ||{\bf x}||^2 \; d{\bf x}\]

A sequence of lattices is said to be {\it Polytrev good} (in terms of channel coding over the AWGN channel) if, for ${\bf \overline{Z}}\sim {\cal N}(0,\overline{\sigma}^2{\bf I})$, we have
\[ \Pr\{{\bf \overline{Z}} \notin {\mathcal V}\} \leq e^{-nE_P(\mu)}, \]
 which upper bounds the error probability of nearest lattice point decoding when using lattice points as codewords in the AWGN channel, and for $E_p$ the {\it Polytrev exponent }\cite{poltyrev1994coding} and \[\mu: = \frac{(\mbox{Vol}({\mathcal V}))^{2/n}}{2\pi e\overline{\sigma}^2}.\]
Since $E_p(\mu) > 0$ for $\mu > 1$, a necessary condition for reliable decoding of a single point is $\mu>1$ - thereby relating the size of the fundamental region (and ultimately how many points one can transmit reliably) to the noise power, aligning well with our intuition about Gaussian noise channels.

\subsection{Nested lattice codes}
\label{subsec:nested}

Now consider two lattices $\Lambda$ and $\Lambda_c$ such that $\Lambda \subseteq \Lambda_c$ with fundamental regions ${\cal V}, {\cal V}_c$ of volumes $V, V_c$ respectively. In this case $\Lambda$ is called the {\it coarse} lattice which is a sublattice of  $\Lambda_c$,  the {\it fine} lattice, and hence $V \geq V_c$.  When transmitting over the AWGN channel, the set  $ \mathcal{C}_{\Lambda_c, {\cal V}} = \{ \Lambda_c \cap \mathcal{V} \} $ is used as the codebook. The coding rate $R$ of this {\it nested  $(\Lambda, \Lambda_c)$ lattice code} is defined as
\[ R = \frac{1}{n} \log |\mathcal{C}_{\Lambda_c, {\cal V}}| = \frac{1}{n} \log \frac{V}{V_c},\] where $\rho = |\mathcal{C}_{\Lambda_c, {\cal V}}|^{\frac{1}{n}} = \left( \frac{V}{V_c} \right)^{\frac{1}{n}}$ is the nesting ratio of the nested lattice. Nested lattice codes were shown to be capacity achieving for the AWGN channel \cite{Erez:2004}.

\subsection{Nested lattice chains}

In the following, we will be using an extension of nested lattice codes termed nested lattice chains, as introduced in \cite{nam:2009nested,  nam:2009bit}, and shown  in Fig. \ref{fig:chain}.
We first re-state  a (slightly modified and simplified)  result of \cite{nam:2009nested} in transmitting codewords over an AWGN channel of transmit power constraint $P$, which will be of use in subsequent sections.

\smallskip

\begin{thm}{\textbf{Existence of ``good'' lattice chains (adapted from Theorem 2 of \cite{nam:2009bit}).}}
\label{thm:nam}
 There exists a sequence of $n$-dimensional lattices $\Lambda \subseteq \Lambda_s \subseteq \Lambda_c $ ( $\mathcal{V}  \supseteq \mathcal{V}_s \supseteq \mathcal{V}_c $) satisfying:  \\
a) $\Lambda$ and $\Lambda_s$ are simultaneously Rogers-good (see pg.7 of \cite{nam:2009nested}) and Poltyrev-good while $\Lambda_c$ is Poltyrev-good.\\
b) For any $\epsilon > 0, P-\epsilon \leq \sigma^2(\Lambda) \leq P $. \\
c) The three rates $ R = \frac{1}{n} \log \frac{V}{V_c} $, $R_1 = \frac{1}{n} \log \frac{V}{V_s} $, $R_2 = \frac{1}{n} \log \frac{V_s}{V_c}$ may  approach any values as $n\rightarrow \infty$, with $R = R_1 + R_2 $, i.e. we have two degrees of freedom in choosing $V,V_s,V_c$ (or equivalently  their second moments as all lattices are Polytrev good).
\end{thm}

\smallskip
This may be derived from Theorem 2 of \cite{nam:2009bit} by noting a one-to-one correspondence between the second moments of $\Lambda$, $\Lambda_s$ and $\Lambda_c$ and their volumes $V, V_s, V_c$ as $n\rightarrow \infty$ as all lattices are Polytrev good; as such the arbitrary second moments of these lattices (of \cite{nam:2009bit}'s Theorem 2) may equivalently be regarded as arbitrary volumes, as in the re-statement above.
Setting $\sigma^2(\Lambda_s)$ to be any value between $0$ and $P$ completes the Theorem.


{The lattice chain result of Theorem 2 in \cite{nam:2009bit} is generalized to a chain of length $K$ in Theorem 2 in \cite{nam:2009nested}; an alternative construction is provided in \cite{nazer2009compute}. We may similarly generalize the result of Theorem \ref{thm:nam} to a chain of  length $K$: for the sequences of $n$-dimensional lattice chains (dimension $n$ left out for simplicity):
$\Lambda_1 \subseteq \Lambda_2 \subseteq \dots \subseteq \Lambda_{K-1} \subseteq \Lambda_K$,
the coding rates of all the nested pairs $R_{ij} = \frac{1}{n} \log \frac{V_i}{V_j} ( 1\leq i \leq j \leq K) $ may approach any values as $n\rightarrow \infty$; we note that due to their definition there are only $K-1$ degrees of freedom in the choice of volumes $V_i$ (or equivalently second moments as all are Polytrev good).

\begin{figure}
\centering
\includegraphics[width=6cm]{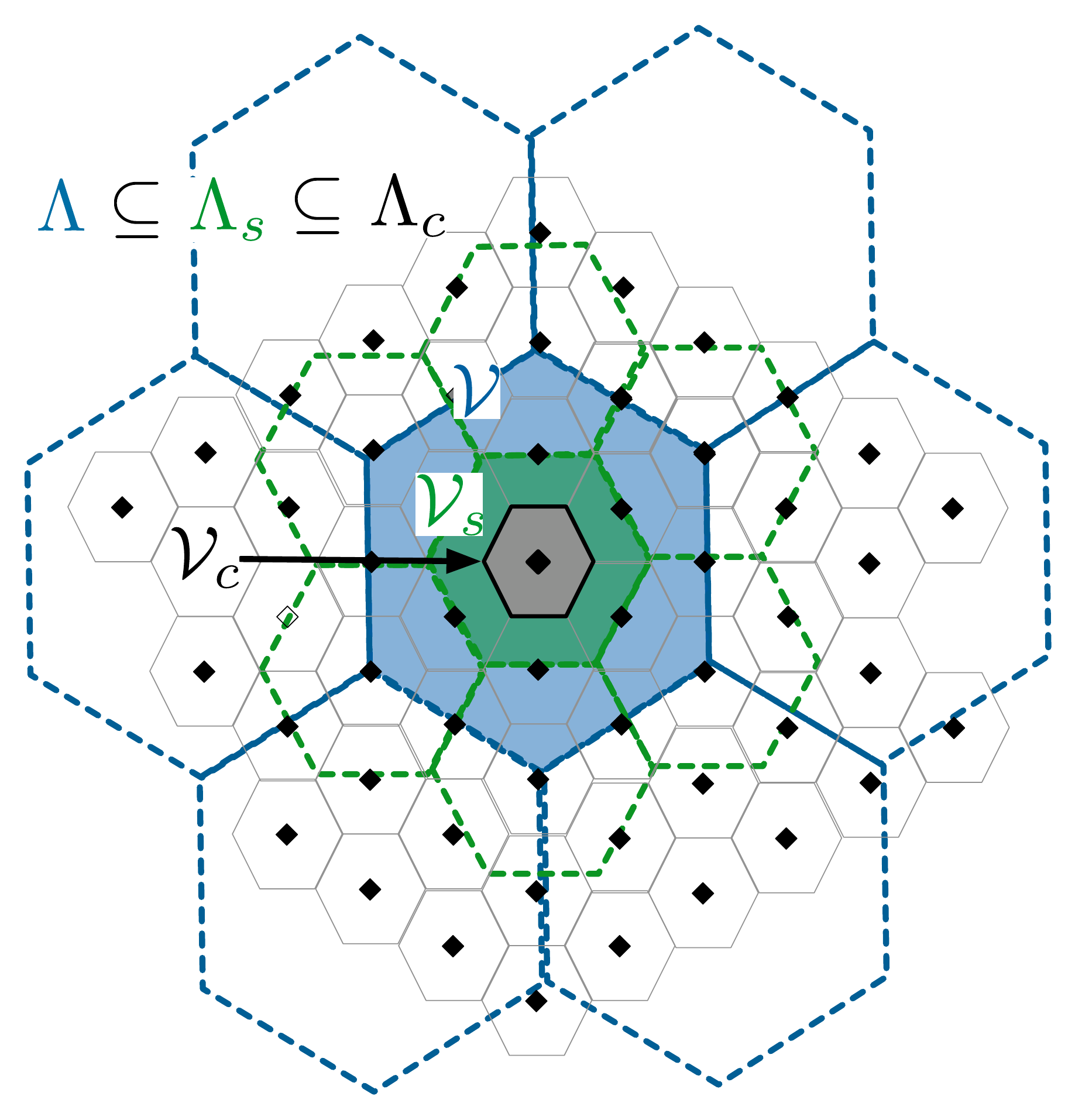}
\caption{A lattice chain $\Lambda\subseteq\Lambda_s \subseteq \Lambda_c$ with corresponding fundamental regions  $\mathcal{V}  \supseteq \mathcal{V}_s \supseteq \mathcal{V}_c $ of volumes $V\geq V_s\geq V_c$. Color is useful.}
\label{fig:chain}
\end{figure}

\subsection{A lattice list decoder}

List decoding refers to a decoding procedure in which, instead of outputting a single codeword corresponding to a single message, the decoder outputs a list of possible codewords which includes the correct (transmitted) one.

We now formalize what we mean by a lattice list decoder for transmitting message of rate $R$ over the AWGN channel ${\bf Y}={\bf X}+{\bf Z}$ where ${\bf Z} \sim {\cal N}(0,N)$, and the channel input ${\bf X}$ is subject to the average power constraint $P$. 
We consider a nested lattice chain  $\Lambda\subseteq\Lambda_s \subseteq \Lambda_c$ as in Section \ref{sec:notation}, Fig. \ref{fig:chain}  and Thm. \ref{thm:nam}.  


\smallskip


{\bf Encoding:} The message set $\{1,\dots,2^{nR}\}$ is one-to-one mapped to $\mathcal{C}_{\Lambda_c, {\cal V}}$. Thus, to transmit a message, the  transmitter chooses the  ${\bf t} \in \mathcal{C}_{\Lambda_c, {\cal V}}$ associated with the message and sends ${\bf X} = ({\bf t} - {\bf U}) \mod \Lambda$, where ${\bf U}$ is a dither signal (known to both the encoder and decoder) uniformly distributed over $\mathcal{V}$.

{\bf Decoding:} Upon receiving {\bf Y}, the receiver computes
\begin{align*}
{\bf Y^\prime} &= (\alpha {\bf Y} + {\bf U}) \mod \Lambda \\
&= ({\bf t} - (1-\alpha){\bf X} + \alpha {\bf Z}) \mod \Lambda \\
&= ({\bf t} + (- (1-\alpha){\bf X} + \alpha {\bf Z} )\mod \Lambda ) \mod \Lambda \\
&= ({\bf t} + {\bf Z^\prime}) \mod \Lambda,
\end{align*}
for $\alpha\in \mathbb{R}$. Choose $\alpha$ to be the MMSE coefficient $\alpha = \frac{P}{P+N} $ and then ${\bf Z^\prime} = (-(1-\alpha){\bf X} + \alpha {\bf Z} ) \mod \Lambda $. Again notice the equivalent noise ${\bf Z^\prime}$ is independent of ${\bf t}$ and $\Lambda_c$.\\
The receiver decodes the {\it list}  of possible codewords 
\begin{equation}
L ({\bf \hat{t}}) := \{{\bf \hat{t}} \;  | \; {\bf \hat{t}} \in  S_{\mathcal{V}_s,\Lambda_c} ({\bf Y^\prime}) \mod \Lambda\},
\label{eq:listdecoding}\end{equation}
where $ S_{\mathcal{V}_s,\Lambda_c}(\overline{\bf x}) := \{ \Lambda_c \cap (\overline{\bf x}+\mathcal{V}_s) \},$
the set of lattice points $\lambda \in \Lambda$ inside the fundamental region ${\cal V}$ centered at the point $\overline{\bf x}$ as shown in Fig. \ref{fig:encdec}.

\smallskip

{\bf Probability of error for list decoding:} In decoding a list, we require that the correct or transmitted codeword lies in the list with high probability as $n\rightarrow \infty$, i.e. the probability of error is $P_e : = \Pr\{{\bf t} \notin L(\bf \hat{t})\}$, which should be made vanishingly small as $n\rightarrow \infty$.

\begin{figure*}
\centering
\includegraphics[width=15cm]{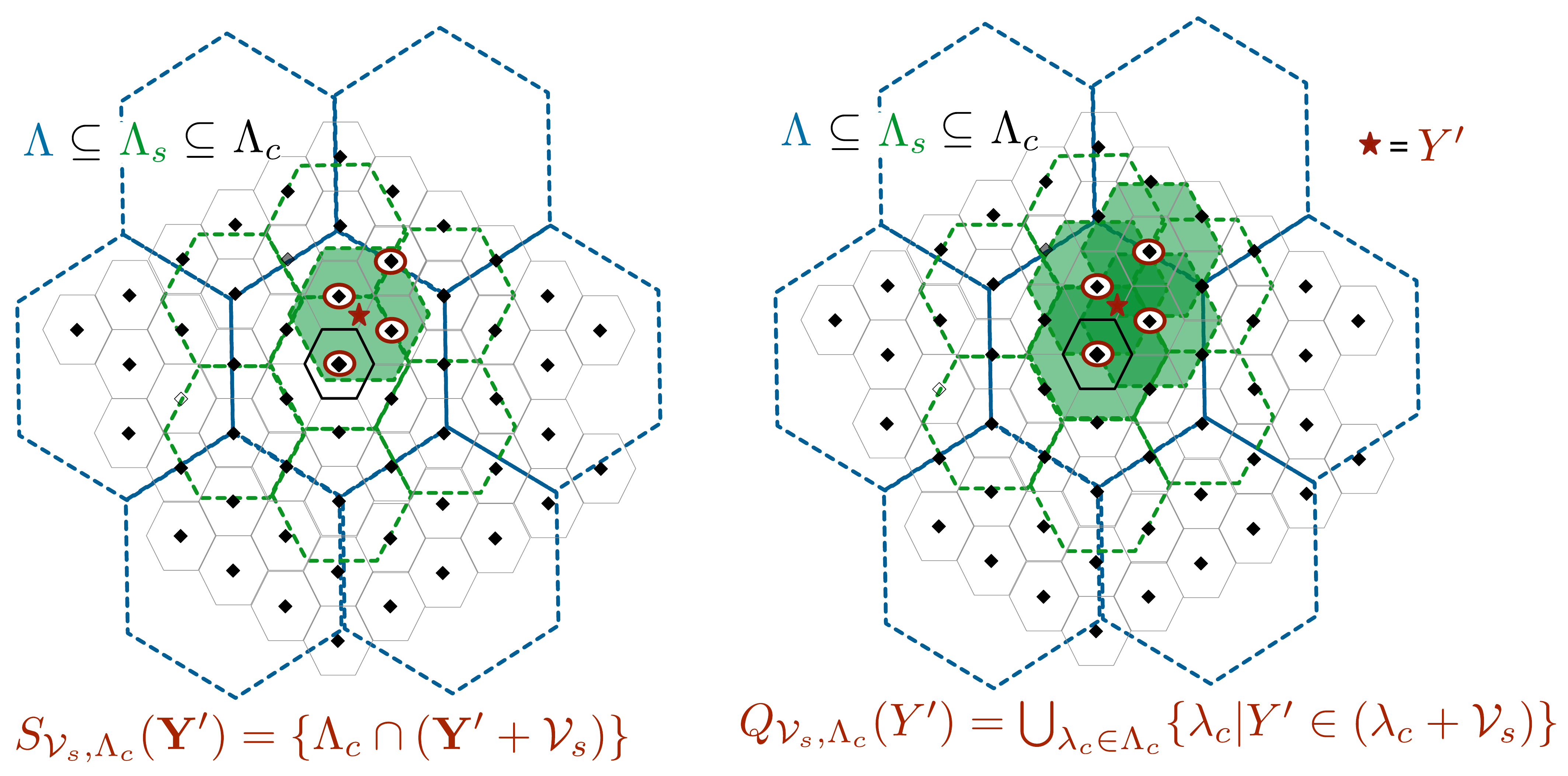}
\vspace{-0.2cm}
\caption{The two equivalent lists, in this example consisting of the four points encircled in red. The correct message  lattice point is the center. Color is useful.}
\label{fig:encdec}
\end{figure*}

\smallskip
\begin{thm}{\textbf{Lattice list decoding.}}
\label{thm:list}
Using the encoding and decoding scheme defined above, the receiver decodes a list of codewords of size   $2^{n(R-C(P/N))}$ with probability of error $P_e \rightarrow 0$ as $n \rightarrow \infty $.
\end{thm}

\begin{proof}
We assume $R > C(P/N)$; When $R \leq C(P/N)$, the decoder can decode an unique codeword with high probability, which was proven in \cite{Erez:2004}.  
In order to prove the above theorem, we will need the following Lemma.

\begin{lemma}{\textbf{Equivalent decoding list.}}
\label{lem:equivalent}
The sets $S_{\mathcal{V}_s,\Lambda_c} ({\bf Y^\prime}) \mod \Lambda$ and $Q_{\mathcal{V}_s,\Lambda_c} ({\bf Y^\prime}) \mod \Lambda$  shown in Fig. \ref{fig:encdec} are equal, where 
\begin{equation}
 Q_{\mathcal{V}_s,\Lambda_c} ({\bf Y^\prime}) := \bigcup_{\lambda_c \in \Lambda_c} \{\lambda_c|{\bf Y^\prime} \in (\lambda_c + \mathcal{V}_s) \}.\label{eq:listdecoding2}
 \end{equation}\end{lemma}
\begin{proof}
We first note that $Q_{\mathcal{V}_s,\Lambda_c} ({\bf Y^\prime}) $ is the set of  $\lambda_c \in \Lambda_c$ points satisfying ${\bf Y^\prime} \in (\lambda_c + \mathcal{V}_s)$. This Lemma will allow us to more easily bound the probability of list decoding error.
First, we note that the fundamental region $\mathcal{V}$ of any lattice $\Lambda$ is centro-symmetric ($\forall x \in \mathcal{V}$, we have that $-x \in \mathcal{V}$) by definition of a lattice and fundamental region (alternatively, see Ch. 4 of \cite{Coppel}). Hence,  for any two points $x$ and $x^\prime$, and a centro-symmetric region $\mathcal{V}$, 
$x^\prime \in x + \mathcal{V} \Leftrightarrow x \in x^\prime + \mathcal{V}$. Applying this to $ S_{\mathcal{V}_s,\Lambda_c}(\overline{\bf x})$ and $Q_{\mathcal{V}_s,\Lambda_c} ({\bf Y^\prime}) $ yields the lemma.
\end{proof}

 We continue with the proof of Thm. \ref{thm:list}. We first (a) use Lemma \ref{lem:equivalent} to see that the lists $S_{\mathcal{V}_s,\Lambda_c} ({\bf Y^\prime}) \mod \Lambda$ and $Q_{\mathcal{V}_s,\Lambda_c} ({\bf Y^\prime}) \mod \Lambda$ are equal. Next we show that
(b) the probability of error of decoding operation $Q_{\mathcal{V}_s,\Lambda_c} ({\bf Y^\prime}) \mod \Lambda$ is
\begin{align}
P_e & = Pr\{ Z^\prime \not \in \mathcal{V}_s \} = e^{ -n( E_p(e^{2(C(P/N) - R_1)}) - o_n(1))} \label{eq:pe}
\end{align}
where $ o_n(1) \rightarrow 0 $ as $n \rightarrow \infty $, $E_p(.)$ is the Poltyrev exponent, and $R_1$ is as defined in Thm. \ref{thm:nam}. 
 Finally, we show that (c) the size of the decoding list is $2^{n(R-C(P/N))}$.

\smallskip

Part (b) follows from our construction of nested lattices according to Thm. \ref{thm:nam} and Theorem 5 of \cite{Erez:2004}. Since $\Lambda$ are both Rogers-good and Poltyrev-good while $\Lambda_s$ is Poltyrev-good, and ${\bf Z^\prime}$ and $\mathcal{V}_s$ are consistent with those in Theorem 5 of \cite{Erez:2004}, all the conditions of Theorem 5 of \cite{Erez:2004} are satisfied and it may thus be applied.

Combining part (a) and (b), we conclude that the probability of error of our decoding operation defined in \eqref{eq:listdecoding} is \eqref{eq:pe}. To ensure $P_e \rightarrow 0$ as $ n \rightarrow \infty $ in \eqref{eq:pe} we need
$ C(P/N) - R_1 > 0, $; recall that
$R_1 = \frac{1}{n} \log ( \frac{V}{V_s} )$.
Combining these, we obtain
\[ V_s > \left(\frac{N}{P+N}\right)^{n/2} \, V.\]
Since $V_c$ may be chosen arbitrarily close to $ \frac{ V}{2^{nR}}$ by Thm. \ref{thm:nam}, the cardinality of the decoded list $L({\bf \hat{t}}) = S_{{\cal V}_s,\Lambda_c} ({\bf Y^\prime}) \mod \Lambda$, in which the true codeword lies with high probability as $n\rightarrow \infty$, may be bounded as 
\begin{align*}
|L ({\bf \hat{t}})| &= \frac{V_s}{V_c} >  \frac{\frac{N^{n/2}V}{(P+N)^{n/2}}}{\frac{ V}{2^{nR}}} = 2^{n(R-C(P/N))}.
\end{align*}
For a given ``good'' lattice chain $\Lambda \subseteq \Lambda_s \subseteq \Lambda_c $ as defined in \ref{thm:nam}, the size of decoded list is fixed (and is not a random variable as in the random-coding based list decoder of \cite{Cover:1979}).
Thus, we may choose  $V_s = \left(\frac{N}{P+N}\right)^{n/2} \, V$, so that the size of decoded list is arbitrarily close to $ 2^{n(R-C(P/N))}$ as $n\rightarrow \infty$.

\end{proof}

\section{Application I of lattice list decoding: the physically degraded Gaussian relay channel}
\label{sec:degraded}

We now show that nested lattice codes may achieve the capacity of the physically degraded relay channel; or the decode-and-forward rates of  Theorem 5 of \cite{Cover:1979}.

\subsection{Channel model}

Consider a  relay channel in which node 1, with channel input $X_{1}$ wishes to transmit a message $w\in \{1,2,\cdots, 2^{nR}\}$ to node
2 which has access to the channel output $Y_{2}$ and is aided by a relay with channel input and output $X_{R}$ and $Y_{R}$.  The physically degraded Gaussian relay channel with transmit power $P$ and relay power $P_R$ is described by 
\begin{align}
{\bf Y_{2}} = {\bf X_{1}} + {\bf X_{R}} + {\bf Z_{2}},  \;\;\;\; {\bf Y_{R}} = {\bf X_{1}} + {\bf Z_{R}},
\end{align}
for ${\bf Z_2} = {\bf Z_R} + {\bf Z_2^\prime} $, where ${\bf Z_2}$ and ${\bf Z_R}$ are sequences
of independent identically distributed Gaussian random variables with mean zero and variances $N_R+N$ and $N_R$ respectively.

The capacity of this Gaussian  physically degraded relay channel was obtained in \cite[Thm. 5]{Cover:1979}; the achievability scheme used includes: (1) random coding, (2) list decoding, (3) Slepian-Wolf partitioning, (4) coding for the cooperative multiple-access channel, (5) superposition coding and (6) block Markov encoding combined with (7) successive decoding. We re-derive this rate region, following the steps of \cite{Cover:1979}'s achievability closely, but with lattice codes replacing the random Gaussian coding techniques.  Of particular importance is the lattice version of the list decoder used on \cite[pg.577]{Cover:1979}.

\medskip

\begin{theorem}
\label{thm:degraded}
The capacity of the physically degraded AWGN relay channel 
may be achieved using nested lattice codes. 
\end{theorem}

\subsection{Proof of Thm. \ref{thm:degraded}; achievability of Gaussian physically degraded relay channel capacity with lattice codebooks}

\smallskip
\noindent

{\bf Construction of the codebooks:} Consider two nested lattice codebooks of dimension $n$ (though we will usually understand all vectors to be of dimension $n$ and will thus drop the dependencies on $n$ in our notation): $(\Lambda_1, \Lambda_{c1})$ and $(\Lambda_2, \Lambda_{c2})$ with $\sigma^2(\Lambda_1) = \alpha P $ and $\sigma^2(\Lambda_2) = \bar{\alpha} P$ $ (\bar{\alpha} = 1 - \alpha)$ for $\alpha\in [0,1]$. By Theorem \ref{thm:nam}, a lattice $\Lambda_{s1}$ also exists such that $\Lambda_1 \subseteq \Lambda_{s1}\subseteq \Lambda_{c1}$ -- this will be used in our list decoder at the destination.
Associate in a 1:1 fashion the message set $\{1,2,\dots,2^{nfR}\}$ with the ${\bf t_1} \in \mathcal{C}_1 = \{ \Lambda_{c1} \cap \mathcal{V}_1\}$, and the message set $\{1,2,\dots,2^{nR_R}\}$ with the set $ {\bf t_2} \in \mathcal{C}_2 = \{ \Lambda_{c2} \cap \mathcal{V}_2\}$.
The messages indices (and corresponding codewords) of the first message set are randomly and uniformly assigned
  indices  in $ \{1,2,\cdots, 2^{nR_R}\}$. The relay, the receiver and the transmitter agree on this assignment.
We use Block Markov coding with successive decoding and define $w_{b}$ as the new message index to be sent in block $b$ ($b=1,2,\cdots, B$); $s_{b}$ is the index corresponding to $w_{b-1}$ in $ \{1,2,\cdots, 2^{nR_R}\}$,  where we define $s_1 = 1$.  It is assumed that at the end of block $b-1$, the receiver knows $(w_1,\dots,w_{b-2} )$ and $(s_1,\dots,s_{b-1})$ and the relay knows $(w_1,\dots,w_{b-1} )$ and $(s_1,\dots,s_{b})$. 

\smallskip
\noindent
{\bf Encoding:} The transmitter sends the superposition (sum) of the codewords ${\bf X_{1}}(w_b) = ({\bf t_1}(w_b) - {\bf U_1}(w_b)) \mod \Lambda_1$ and ${\bf X_{2}}(s_b) = ({\bf t_2}(s_b) - {\bf U_2}(s_b)) \mod \Lambda_2$. The relay sends ${\bf X_{R}}(s_b) = \sqrt { \frac{P_R}{\bar{\alpha}P}} {\bf X_2}(s_b) = (\sqrt { \frac{P_R}{\bar{\alpha}P}} {\bf t_2}(s_b) - \sqrt { \frac{P_R}{\bar{\alpha}P}}{\bf U_2}(s_b) ) \mod \sqrt { \frac{P_R}{\bar{\alpha}P}} \Lambda_2 $, for ${\bf U_1}(w_b)$ and ${\bf U_2}(s_b)$ dithers known to all nodes which are i.i.d. and also change from block to block (which we indicate, with some abuse of notation, by a dependence on $s_b$ and $w_b$). 
\smallskip

\noindent
{\bf Decoding:}

1. The relay knows $s_b$ and consequently ${\bf X_2}(s_b)$, and so may decode the message $w_b$ from the received signal $ {\bf Y_{R}} = {\bf X_1}(w_b) + {\bf X_2}(s_b) + {\bf Z_{R}} $ as long as $ R < C(\alpha P/N_R)$,
since the ``good'' nested lattice code pair $(\Lambda_1, \Lambda_{c1})$  can achieve the capacity of the point-to-point channel \cite{Erez:2004}.

2. The receiver can decode $s_b$ from the signal $ {\bf Y_2} = {\bf X_1}(w_b) + {\bf X_2}(s_b) + {\bf X_R}(s_b) + {\bf Z_2}$ as long as
\[ R_R < \frac{1}{2} \log \left( 1 + \frac{ (\sqrt{\bar{\alpha}P} + \sqrt{P_R})^2 }{ \alpha P + N + N_R } \right)\]
since
\begin{align*}
{\bf Y_2} &= {\bf X_1} + {\bf X_2} + {\bf X_R} + {\bf Z_2} \\
&= {\bf X_1} + \left(1 +\sqrt { \frac{P_R}{\bar{\alpha}P}} \right) {\bf X_2} + {\bf Z_2}.
\end{align*}
Now define $\kappa: = (1+\sqrt{P_R/(\bar{\alpha}P_R)})$. Then, ${\bf t^\prime_2} = \kappa {\bf t_2}$, ${\bf U^\prime_2} =\kappa {\bf U_2}$, and $\Lambda^\prime_2 =\kappa  \Lambda_2$, and finally ${\bf X^\prime_2} = \kappa {\bf X_2}$.
Thus $ {\bf Y_2} = {\bf X_1} + {\bf X_2^\prime} + {\bf Z_2} $, and so we may compute
\begin{align*}
{\bf Y^\prime} &= ( \beta {\bf Y} + {\bf U_2}^\prime  ) \mod \Lambda_2^\prime \\
&= ({\bf t_2}^\prime - (1 - \beta) {\bf X_2^\prime} + \beta({\bf X_1} + {\bf Z}) ) \mod \Lambda_2^\prime.
\end{align*}
Choose an appropriate lattice pair $(\Lambda_2, \Lambda_{c2})$ so that $(\kappa \Lambda_2, \kappa \Lambda_{c2})$ (i.e. $(\Lambda_2^\prime, \Lambda_{c2}^\prime)$) is  a ``good'' nested lattice pair  \cite{Erez:2004}. 
  We notice that $\sigma^2(\Lambda_2^\prime) = P^\prime = \kappa^2 \bar{\alpha}P$. 
 As in \cite{Erez:2004}, choose $\beta$ to be the MMSE coefficient $\beta = \beta_{MMSE}= \frac{P^\prime}{P^\prime + \alpha P + N + N_R}$, resulting in the equivalent self-noise of variance
 \[N_{eq} = \frac{P^\prime (\alpha P + N + N_R ) } { P^\prime + \alpha P + N + N_R }.\]
 Thus, ${\bf t_2}^\prime$ and then ${\bf t_2}$ and $s_b$ may be decoded as long as (see  \cite{nazer2009compute} for details of this decoding step and error analysis):
\begin{align*}
R_R &< \frac{1}{2} \log \left( \frac{ P^\prime }{ \frac{P^\prime (\alpha P + N + N_R ) } { P^\prime + \alpha P + N + N_R } } \right) \\
&= \frac{1}{2} \log \left( 1 + \frac{ (\sqrt{\bar{\alpha}P} + \sqrt{P_R})^2 }{ \alpha P + N + N_R } \right).
\end{align*}

3. The receiver now subtracts ${\bf X^\prime_2}$ from ${\bf Y_2}$: ${\bf Y_2} - {\bf X^\prime_2}(s_b) = {\bf X_1}(w_b) + {\bf Z_2}$, and decodes a list  of possible codewords ${\bf t_1}(w_b)$ of  size $2^{n(R - C(\alpha P/(N + N_R)))}$ by the lattice list decoding scheme shown in Section \ref{subsec:nested} (and using the nested lattice chain  $\Lambda_1 \subseteq \Lambda_{s1}\subseteq \Lambda_{c1}$ ).  Here, we choose the nested list decoding lattice $\Lambda_{s1}$ (middle one) to have a fundamental region is of volume $V_{s1} = \left( \frac{N + N_R}{ \alpha P + N + N_R} \right)^{n/2} V_1$ {asymptotically} so that the size of the decoded list is $2^{n(R-C(\alpha P/(N + N_R)))}$.
 This list of messages denoted by $L(\hat{w}_b)$ and will be used in the next block (block $b+1$). To decode $w_{b-1}$, the receiver  intersects the decoded partition $\hat{s}_b$ (which includes a group of possible message indices $\hat{w}_{b-1}$) with the list of possible messages $L(\hat{w}_{b-1})$, and declares a success if there is a unique $w_{b-1}$  in this intersection. Due to the uniform and random partitioning of message indices into the $2^{nR_R}$ ``bins'', this is the case if $R - C(\alpha P/(N+N_R))< R_R$, or
\begin{align*} R & < \frac{1}{2} \log \left( 1 + \frac{\alpha P}{N + N_R} \right) + R_R \\
&< \frac{1}{2} \log \left( 1 + \frac{ P + P_R + 2\sqrt{\bar{\alpha} P P_R} } { N + N_R} \right).
\end{align*}

\begin{figure*}
\centering
\includegraphics[width=14cm]{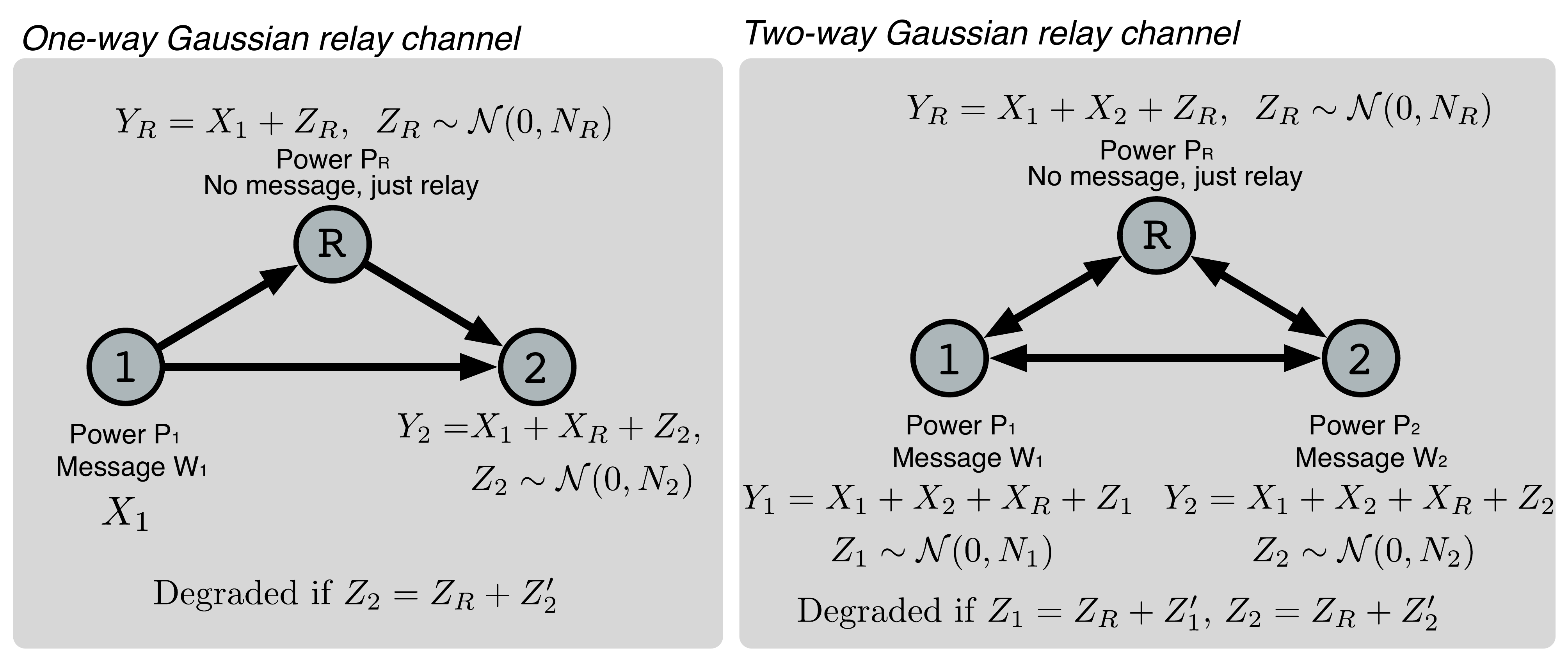}
\caption{The two Gaussian relay channels under consideration in Section \ref{sec:degraded} and Section  \ref{sec:two-way}. Note that each receiver may subtract off its own message.}\label{fig:relay-channel}
\end{figure*}

\section{The two-way relay channel with direct links}
\label{sec:two-way}

The two-way relay channel is the logical extension of the classical relay channel \cite{Cover:1979} for one-way point-to-point  communication aided by a relay to allow for two-way communication.
While the capacity region is in general unknown, it is known for half-duplex channel models under the 2-phase MABC protocol  \cite{Kim:sarnoff} and to within 1/2 bit for the full-duplex Gaussian channel model with no direct links  \cite{nam:2009bit}, and to within 2 bits for the same model with direct links in certain ``symmetric'' cases \cite{avestimehr2009capacity}.

\subsection{Prior achievable rate regions with lattices}

While random coding techniques employing a DF, CF, and AF relays have been the most common in deriving achievable rate regions for the two-way relay channel, a handful of work \cite{Baik:2008, nam:2009bit, Narayanan:2008, ong2010capacity} has considered lattice-based schemes which, in a DF-like setting,  effectively exploit the additive nature of the the Gaussian noise channel in allowing the ``sum''  of the two transmitted lattice points to be decoded at the relay. The intuitive gains of decoding the sum of the messages rather than the individual messages stem from the absence of the classical multiple-access constraints.  This sum-rate point is forwarded to the terminal (which may be re-encoded using a random or lattice code), which utilizes its own-message side-information to ``subtract'' off its own message from the decoded sum.

While random coding schemes have been used in deriving achievable rate regions in the presence of direct links, lattice codes -- of interest in order to exploit the ability to decode the sum of messages at the relay -- have so far not been used. 
 We attack this next using a random binning technique at the relay \cite{xie2007network}, combined with lattice list decoding at the end users.

\subsection{Channel model}

Our channel model consists of two terminal nodes with inputs $X_1, X_2$ with power constraints $P_1, P_2$ (without loss of generality, it is assumed $P_1 \geq P_2$) and outputs $Y_1, Y_2$ which wish to exchange messages with the help of the relay with input $X_R$ of power $P_R$ and output $Y_R$. We assume a memoryless AWGN channel model with direct links between the terminals, described by the input/output relations at each channel use (we drop the index $i$ for simplicity)
\begin{align*}
{\bf Y_1} &= {\bf X_R} + {\bf X_2} + {\bf Z_1}, \;\;\;\; {\bf Z_1} \sim {\cal N}(0,N_1) \\
{\bf Y_2} &= {\bf X_R} + {\bf X_1} + {\bf Z_2}, \;\;\;\; {\bf Z_2} \sim {\cal N}(0,N_2) \\
{\bf Y_R} &= {\bf X_1} + {\bf X_2} + {\bf Z_R}, \;\;\;\; {\bf Z_R} \sim {\cal N}(0,N_R)
\end{align*}
(due to the additive nature of AWGN channel models, we drop the transmitters' own signal for simplicity )
with average input power constraints $P_1,P_2, P_R$, respectively. Note that while the channel gains appear to all be identical, this channel model may be assumed without loss of generality as we allow for arbitrary noise and input powers.
 We call this two-way relay channel {\it physically degraded} if ${\bf Z_1} ={\bf Z_R} +{\bf Z_1'}$ (${\bf Z_1^\prime} \sim {\cal N}(0,N_1^\prime)$) and ${\bf Z_2} = {\bf Z_R} + {\bf Z_2'}$ (${\bf Z_2^\prime} \sim {\cal N}(0,N_2^\prime)$); and {\it stochastically degraded} if $N_1, N_2 \geq N_R$.  

\subsection{A new achievable rate region for the Gaussian two-way relay channel with direct links, lattice codes and list decoding}

In two-way communications (with restricted terminal nodes whose transmissions may not depend on past channel outputs) in the presence of a relay, the relay may be used or ignored in either direction, leading, without loss of generality, to three possible cases: (1) both directions ignore the relay;  (2) both directions use the relay; and  (3) one direction uses the relay while the other ignores it.
Case (1) results in the capacity region of the AWGN two-way channel \cite{Han:1984} given by
\begin{align}
R_i &\leq \frac{1}{2} \log \left( 1 + \frac{P_i}{N_{\bar{i}}} \right), \;\;\; i=1,2. \label{eq:R11}
\end{align}
 We derive a new region for case (2), and leave (3) for future work. 
We note that constant gaps are known for scenarios in which the first two-cases are useful \cite{avestimehr2009capacity}.


\medskip

\begin{theorem}
\label{thm:two-way}
The following rates are achievable for the two-way AWGN relay channel with direct links

{\footnotesize \begin{align}
R_1 & \leq \min \left( \left[\frac{1}{2} \log \left( \frac{P_1}{P_1 + P_2} + \frac{P_1}{N_R}\right)\right]^+ , \frac{1}{2} \log\left(1+\frac{P_1+P_R}{N_2}\right)\right) \label{eq:R21}\\
R_2 &\leq  \min \left( \left[\frac{1}{2} \log \left( \frac{P_2}{P_1 + P_2} + \frac{P_2}{N_R}\right)\right]^+ , \frac{1}{2} \log\left(1+\frac{P_2+P_R}{N_1}\right)\right) \label{eq:R22}.
\end{align}}
\end{theorem}

\begin{proof}
We use  the random binning technique developed by \cite{xie2007network},  nested lattice codes at the terminals,
and the lattice list decoding scheme of Section \ref{sec:lattice}.

{\it Codebook generation:}
By the generalization of Thm. \ref{thm:nam}, which is proved in \cite{nam:2009nested}, there exists a chain of lattices $\Lambda_1 \subseteq \Lambda_2 \subseteq \Lambda_{c1}  \subseteq \Lambda_{c2} $ (or $\Lambda_1 \subseteq \Lambda_2 \subseteq \Lambda_{c2}  \subseteq \Lambda_{c1}$, $\Lambda_1 \subseteq \Lambda_{c1} \subseteq \Lambda_{2}  \subseteq \Lambda_{c2}$: the order of this lattice chain sequence depends on $V_1$, $V_2$, $V_{c1}$, and $V_{c2}$ and consequently on $P_1$, $P_2$, $R_1$ and $R_2$), where $\Lambda_1$ and $\Lambda_2$ are Rogers-good and Poltyrev-good, while $\Lambda_{c1}$ and $\Lambda_{c2}$ are Poltyrev-good and $\sigma^2(\Lambda_1) = P_1$, $\sigma^2(\Lambda_2) = P_2$.
We also note that we may construct / there exist additional lattices $\Lambda_{s1}$ and $\Lambda_{s2}$ that are appropriately nested in order to invoke list decoding at both receivers, i.e.  $\Lambda_1 \subseteq \Lambda_{s1}\subseteq \Lambda_{c1}$  and  $\Lambda_2 \subseteq \Lambda_{s2}\subseteq \Lambda_{c2}$, which will form a lattice chain with 6 lattices, whose order depends on the relative strengths of the channel links.
For terminal 1, associate each message $w_1 \in \{1,\dots,2^{nR_1}\}$ with ${\bf t_1} \in \mathcal{C}_1 = { \Lambda_{c1} \cap \mathcal{V}_1 }$. For terminal 2, associate each message $w_2 \in \{1,\dots,2^{nR_2}\}$ with ${\bf t_2} \in \mathcal{C}_2 = { \Lambda_{c2} \cap \mathcal{V}_2 }$. For the relay, independently generate $2^{nR}$ $n$-sequences $X_R^n$ with components generated i.i.d. according to the Gaussian distribution with mean 0 and variance $P_R$, for $ R\geq \max(I(X_R;Y_2|X_2), I(X_R;Y_1|X_1))$, similar to the type of binning performed in \cite{xie2007network} for broadcasting information with receiver side-information.

{\it Encoding:} Messages $w_{1b}$ and $w_{2b}$ are the messages the two terminals want to send in the block $b$. We use a block Markov transmission strategy where in the $b$-th block, terminal 1 sends ${\bf X_1}(w_{1b}) = ({\bf t_1}(w_{1b}) - {\bf U_1}(w_{1b})) \mod \Lambda_1$, and terminal 2 sends ${\bf X_2}(w_{2b}) = ({\bf t_2}(w_{2b}) - {\bf U_2}) \mod \Lambda_2 $ for dithers ${\bf U_1}, {\bf U_2}$ known to all nodes (which are iid over channel uses and vary from block to block). At the relay,  we assume that it has decoded ${\bf \hat{T}}(b-1) = ({\bf t_1}(w_{1(b-1)}) + {\bf t_2}(w_{2(b-1)}) - Q_2({\bf t_2}w_{2(b-1)}) + {\bf U_2}(w_{2(b-1)}))) \mod \Lambda_1 $ in block $b-1$. ${\bf \hat{T}}$ is thrown uniformly (or binned) into the $2^{nR}$ bins, and $s({\bf \hat{T}})$ is defined as ${\bf \hat{T}}$'s bin index. Terminal 3 sends ${\bf X_R}(s({\bf \hat{T}}(b-1)))$.

{\it Decoding:} At the end of each block $b$, the relay terminal can decode ${\bf T}(b) = ({\bf t_1}(w_{1b}) + {\bf t_2}(w_{2b}) - Q_2({\bf t_2}(w_{2b}) + {\bf U_2}(w_{2b}))) \mod \Lambda_1 $ as long as

{\footnotesize \begin{align*}
R_1 &\leq \frac{1}{2} \log \left( \frac{P_1}{P_1 + P_2} + \frac{P_1}{N_R}\right),  \;\; R_2 \leq \frac{1}{2} \log \left( \frac{P_2}{P_1 + P_2} + \frac{P_2}{N_R}\right).
\end{align*}}

\vspace{-0.1in}
This follows from arguments similar to those in \cite{nazer2009compute, nam:2009bit}.

We now consider the decoding of message $w_{1(b-1)}$ at terminal 2 after block $b$, which closely follows the backwards decoding strategy of the one-way relay channel \cite{Cover:1979}. That is, after block $b-1$ terminal 2 used the list decoder of Section \ref{sec:lattice} to produce a list of possible ${\bf t_1}(w_{1(b-1)})$, say $L(\hat{{\bf t_1}}(w_{1(b-1)}))$ of size $2^{n(R_1 - C(P_1/N_2))}$. To resolve which codeword in this list was actually sent, it uses the bin-index it receives in block $b$ from the relay, $s({\bf \hat{T}}(b-1))$. To decode this bin index, we use Xie's random binning scheme \cite{xie2007network}.  Note that given ${\bf U_1}$ and ${\bf U_2}$, for fixed ${\bf t_1}$,  ${\bf t_2}$ and ${\bf T}$ are in one-to-one correspondence, while for fixed ${\bf t_2}$,   ${\bf t_1}$ and ${\bf T}$ are in one-to-one correspondence:
\[ {\bf T} = ({\bf t_1} + {\bf t_2} - Q_2({\bf t_2} + {\bf U_2})) \mod \Lambda_1 \]
and
{\small \begin{align*}
(&{\bf T} - {\bf t_2} + Q_2({\bf t_2} + {\bf U_2}) )\mod \Lambda_1 \\
&= ( ({\bf t_1} + {\bf t_2} - Q_2({\bf t_2} + {\bf U_2})) - {\bf t_2} + Q_2({\bf t_2} + {\bf U_2}) ) \mod \Lambda_1 \\
&= {\bf t_1} \mod \Lambda_1 ={\bf t_1}
\end{align*}}
and
{\small \begin{align*}
(&{\bf T} \mod \Lambda_2 - {\bf t_1}) \mod \Lambda_2 \\
&= ( ({\bf t_1} + {\bf t_2} - Q_2({\bf t_2} + {\bf U_2}) ) \mod \Lambda_1 \mod \Lambda_2 - {\bf t_1} )\mod \Lambda_2 \\
&=  ( ({\bf t_1} + {\bf t_2} - Q_2({\bf t_2} + {\bf U_2}) ) \mod \Lambda_2 - {\bf t_1} )\mod \Lambda_2 \\
&= {\bf t_2} \mod \Lambda_2 = {\bf t_2}.
\end{align*}}
The second equality follows from ${\bf X}\mod \Lambda_1 \mod \Lambda_2 = {\bf X} \mod \Lambda_2 $ when $\Lambda_1\subseteq \Lambda_2$.
Thus, since terminal 2 knows $w_{2(b-1)}$ and consequently ${\bf t_2}(w_{2(b-1)})$, terminal 2 decodes the unique ${\bf t_1}(w_{1(b-1)})$ which in block $b$ satisfies the joint typicality check
\[ ( {\bf x_R}(s({\bf T}(b-1))), {\bf X_2}(w_b), {\bf Y_2}(b) ) \in A^{(N)}_\epsilon ({\bf X_R}, {\bf X_2}, {\bf Y_2})\]
and also in block $b-1$ belongs to the list of possible codewords $L(\hat{{\bf t_1}}(w_{1(b-1)}))$ of size $2^{n(R_1 - C(P_1/N_2))}$.
Due to the uniform, random binning performed to obtain the bin index of ${\bf \hat{T}}$, this is possible as long as
 $R_1 < I(X_R;Y_2|X_2) + C(P_1/N_2)$.
Since ${\bf Y_2} = {\bf X_R} + {\bf X_1} + {\bf Z_2}$ and the distribution of ${\bf X_1}$ (lattice code $\Lambda_1$ which is Rogers-good) approaches the Gaussian distribution with variance $P_1$ as $n \rightarrow \infty$ \cite{Erez:2004, loeliger1997averaging},
\[ I(X_R;Y_2|X_2) = \frac{1}{2} \log \left( 1 + \frac{P_R}{P_1+N_2}\right).\]
Thus,
\begin{align*}
R_1 &< I(X_R;Y_2|X_2) + C(P_1/N_2)\\
&= \frac{1}{2} \log \left( 1 + \frac{P_R + P_1}{N_2}\right).
\end{align*}
Analogous steps apply to rate $R_2$.
\end{proof}






\subsection{An improved partial finite-gap result}

The capacity region of the two-way relay channel  is known to within 1/2 bit \cite{nam:2009bit} without direct links, and to within 2 bits for some specific cases (when either the relay is better or worse than the direct links in {\it both} directions) when direct links are present \cite[Section V]{avestimehr2009capacity} per user. We note that recent ``noisy network coding'' \cite{lim2010noisy} techniques may also lead to constant gaps for this channel with direct links but to the best of our knowledge such gaps have not been published. We improve upon the partial constant gap results of \cite{avestimehr2009capacity} under similar channel conditions. That is, in the ``interesting case'' of  \cite{avestimehr2009capacity} where the direct links are weaker than the relay paths, we improve the 2 bit gap to $\frac{1}{2}\log 3$ bits when the relay is a better receiver than the two destinations (i.e. Scenario (2) in the following). Notice that Scenario (2) includes and extends upon the ``interesting'' case of \cite{avestimehr2009capacity}.
Both directions cooperatively use the direct links and relay in Scenarios (1) and (2), (eqns. \eqref{eq:R21} -- \eqref{eq:R22}); we note that, while not included here for lack of space, constant gaps are also available for reversely degraded cases (physical and stochastic) by ignoring the relay. 
We indicate the constant gaps (to the cut-set outer bound).


\medskip

$\bullet$ Scenario (1): Two-way physically degraded case, i.e. ${\bf Z_1} = {\bf Z_R} + {\bf Z_1^\prime}$ AND ${\bf Z_2} = {\bf Z_R} + {\bf Z_2^\prime}$: $\frac{1}{2}$ bit gap.

$\bullet$ Scenario (2): Two-way stochastically degraded case, i.e.  $N_1 \geq N_R$ AND $N_2 \geq N_R$:  $\frac{1}{2} \log 3$ bit gap.

\medskip




Notice that all above scenarios involve some form of symmetry in the channel conditions, essentially implying that it it either optimal to use or not use the relay in both directions.  In asymmetric scenarios, such as  for example: $N_2 \geq N_R$ AND $\min (  \frac{P_2}{N_R}, \frac{P_R}{N_1} ) \leq \frac{P_2}{N_1}$, our scheme cannot achieve within a finite gap of the outer bound because either our new rate region (eqns. \eqref{eq:R21} -- \eqref{eq:R22}) or that of the two-way AWGN channel \eqref{eq:R11} can only guarantee a finite gap in one direction ($R_1$ OR $R_2$). In particular, one disadvantage of our scheme (2) is that as both directions use the relay, the sum of the messages has to be decoded, leading to constraints on both rates $R_1$ and $R_2$, even if one link has a better direct link and wishes to ignore the relay.  The cut-set outer bound on the other hand permits the direct link to be fully exploited. 



\bigskip

\noindent
{\bf Scenario (1):}
The cut-set bound for the physically degraded Gaussian two-way relay channel is given by \eqref{eq:outd2}.
\begin{figure*}
\begin{equation}
R_i \leq R_{ODi} = \max_{ 0 \leq \alpha_i \leq 1} \min\left( \frac{1}{2} \log \left( 1 + \frac{\alpha_i P_i}{N_R} \right) \right., 
 \left.  \frac{1}{2} \log \left( 1 + \frac{P_i + P_R + 2 \sqrt{\bar{\alpha_i} P_i P_R} }{N_{\bar{i}}^\prime + N_R} \right) \right), \;\; i\in \{1,2\} . 
\label{eq:outd2} \end{equation} 
\end{figure*}
It is shown here that the achievable rates of Thm. \ref{thm:two-way} lie within 1/2 bit of this outer bound (per user). Note that

{\footnotesize
\begin{align*}
 \eqref{eq:R21} & + \frac{1}{2} \\
= & \min \left( \left[\frac{1}{2} \log \left( \frac{P_1}{P_1 + P_2} + \frac{P_1}{N_R}\right)\right]^+ , \frac{1}{2} \log\left(1+\frac{P_1+P_R}{N_2^\prime +N_R}\right)\right) + \frac{1}{2} \\
= & \min \left( \max \left( \frac{1}{2} \log \left( \frac{2P_1}{P_1 + P_2} + \frac{2P_1}{N_R}\right), \frac{1}{2}\right) , \frac{1}{2} \log\left(2+\frac{2(P_1+P_R)}{N_2^\prime +N_R}\right)\right),
\end{align*}
}
and that both terms are lower bounded by $R_{OD1}$, as

{\footnotesize\begin{align*}
& \max \left( \frac{1}{2} \log \left( \frac{2P_1}{P_1 + P_2} + \frac{2P_1}{N_R}\right), \frac{1}{2}\right) \geq \max \left( \frac{1}{2} \log \left( \frac{2P_1}{N_R}\right), \frac{1}{2}\right) \\
& \geq \frac{1}{2} \log \left( 1 + \frac{ P_1}{N_R} \right)
\geq \frac{1}{2} \log \left( 1 + \frac{\alpha_1 P_1}{N_R} \right) \geq R_{OD1}\\
\end{align*}
and
\begin{align*}
&\frac{1}{2} \log\left(2+\frac{2(P_1+P_R)}{N_2^\prime +N_R}\right) > \frac{1}{2} \log \left( 1 + \frac{P_1 + P_R + 2 \sqrt{ P_1 P_R} }{N_2^\prime + N_R} \right) \\
&\geq \frac{1}{2} \log \left( 1 + \frac{P_1 + P_R + 2 \sqrt{\bar{\alpha}_{1} P_1 P_R} }{N_2^\prime + N_R} \right)
\geq R_{OD1}.
\end{align*}
}

Thus,
$ \eqref{eq:R21} + \frac{1}{2}
\geq  R_{OD1}.$
A similar $1/2$ bit gap results for user 2's rate $\eqref{eq:R22}$. 
Scenario (2) follows in a similar manner; we note that the cut-set outer bound is no longer given by \eqref{eq:outd2}. 

\begin{figure*}
\subfigure{\includegraphics[width=2.4in]{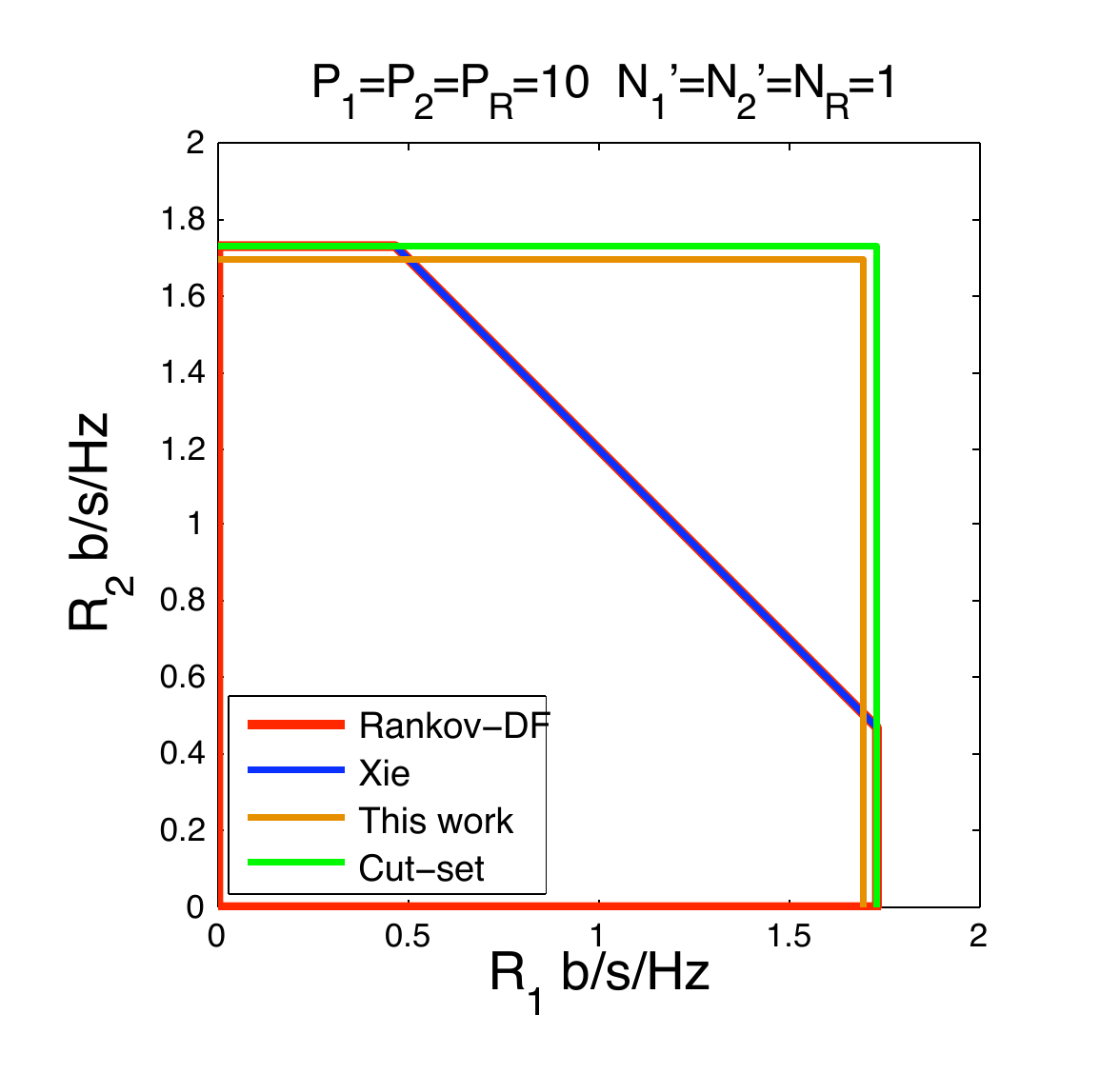}}
\subfigure{\includegraphics[width=2.4in]{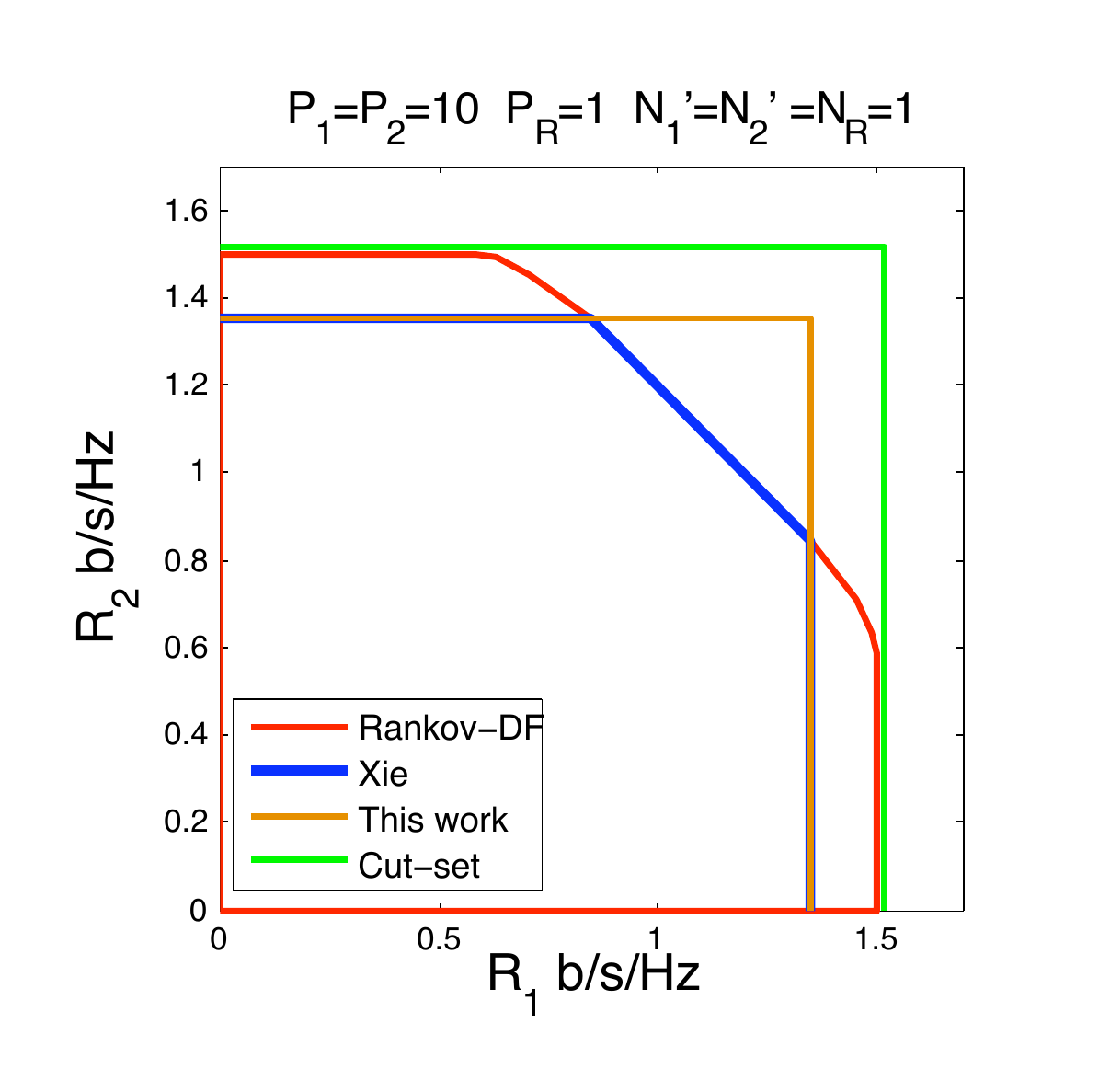}}
\subfigure{\includegraphics[width=2.4in]{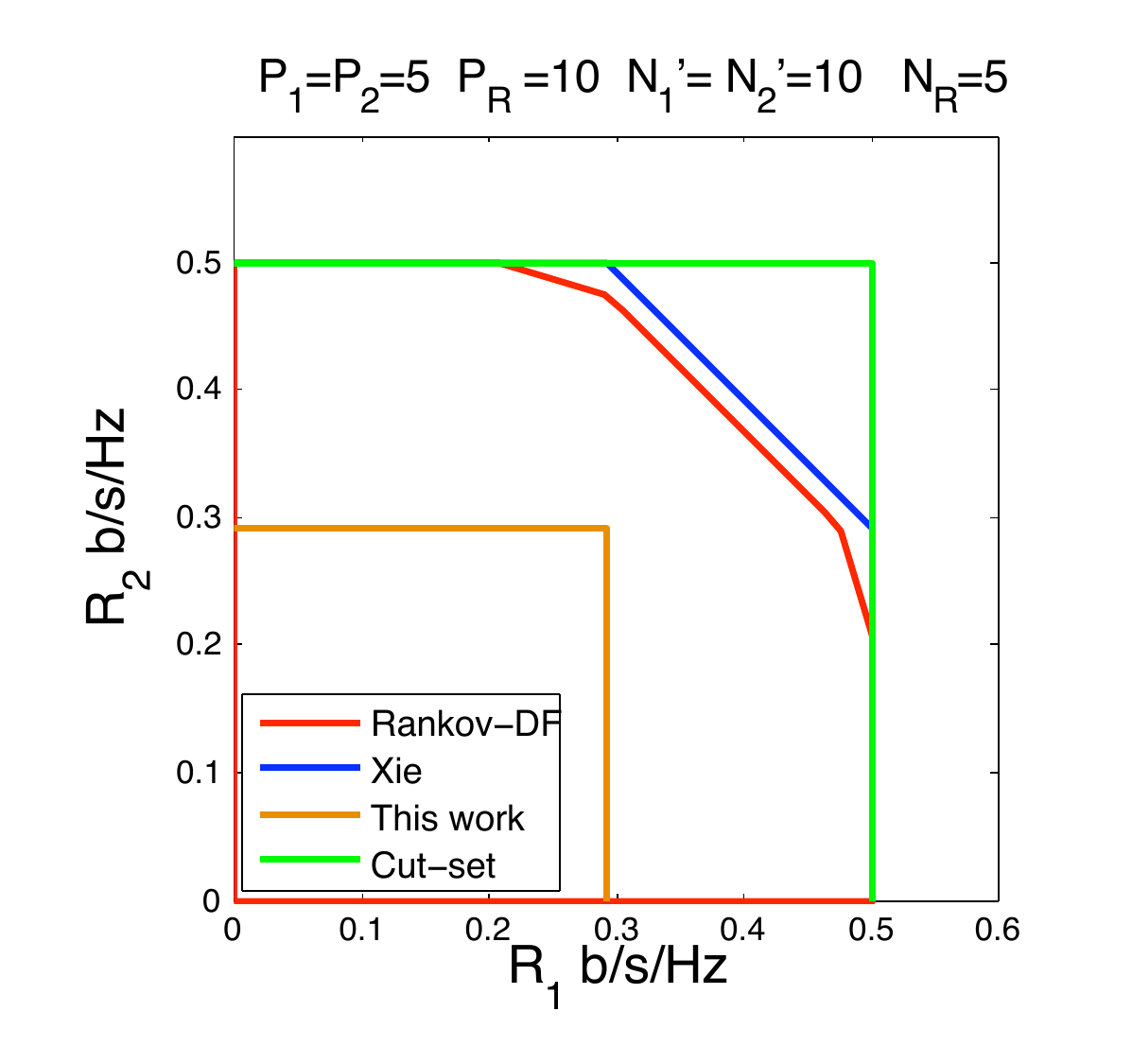}}
\caption{Comparison of decode-and-forward achievable rate regions of various two-way relay channel rate regions.}
\label{fig:numerical}
\end{figure*}

\subsection{Numerical evaluations}

We compare three achievable rate regions of decode-and-forward (DF) schemes with direct links to the cut-set outer bound in Fig. \ref{fig:numerical} for the degraded channel (in which scheme (2) is most useful): the red ``Rankov-DF'' \cite{Rankov:2006}, the blue ``Xie'' \cite{xie2007network} and our orange ``This work'' (Thm. \ref{thm:two-way}). The ``Rankov-DF'' and ``Xie'' schemes use a multiple access channel model to decode the two messages at the relay, while we use lattice codes to decode their sum, which avoids the sum rate constraint. In the broadcast phase, the ``Rankov-DF'' scheme broadcasts the superposition of the two codewords,
while ``Xie'''s and our scheme use a random binning technique to broadcast the bin index. The advantage of  the ``Rankov-DF'' scheme is its ability of obtain a coherent gain at the receiver from the source and 
relay at the cost of a reduced power for each message (power split $\alpha P$ and $(1-\alpha)P$).  Xie and our schemes both broadcast the bin index using all of the relay power.  At low SNR, the rate-gain seen by decoding the sum and avoiding the sum-rate constraint is outweighed by 1) loss seen in the rates $\frac{1}{2} \log ( \frac{P_i}{P_1 + P_2} + SNR)$ compared to $\frac{1}{2} \log ( 1 + SNR)$, or 2) the coherent gain present in the ``Rankov-DF'' scheme. At high SNR, our scheme performs well, and at least in some cases, is able to guarantee a constant gap.

\smallskip

{\bf Acknowledgements.} 
The authors would like to thank Bobak Nazer for his comments on a draft of this manuscript.
\bibliographystyle{IEEEtran}
\bibliography{refs}

\begin{thebibliography}{10}
\providecommand{\url}[1]{#1}
\csname url@rmstyle\endcsname
\providecommand{\newblock}{\relax}
\providecommand{\bibinfo}[2]{#2}
\providecommand\BIBentrySTDinterwordspacing{\spaceskip=0pt\relax}
\providecommand\BIBentryALTinterwordstretchfactor{4}
\providecommand\BIBentryALTinterwordspacing{\spaceskip=\fontdimen2\font plus
\BIBentryALTinterwordstretchfactor\fontdimen3\font minus
  \fontdimen4\font\relax}
\providecommand\BIBforeignlanguage[2]{{%
\expandafter\ifx\csname l@#1\endcsname\relax
\typeout{** WARNING: IEEEtran.bst: No hyphenation pattern has been}%
\typeout{** loaded for the language `#1'. Using the pattern for}%
\typeout{** the default language instead.}%
\else
\language=\csname l@#1\endcsname
\fi
#2}}

\bibitem{Erez:2004}
U.~Erez and R.~Zamir, ``Achieving $\frac{1}{2}\log(1+{SNR})$ on the
  {A}{W}{G}{N} channel with lattice encoding and decoding,'' \emph{IEEE Trans.
  Inf. Theory}, vol.~50, no.~10, pp. 2293--2314, Oct. 2004.

\bibitem{zamir2002nested}
R.~Zamir, S.~Shamai, and U.~Erez, ``{Nested linear/lattice codes for structured
  multiterminal binning},'' \emph{IEEE Transactions on Information Theory},
  vol.~48, no.~6, pp. 1250--1276, 2002.

\bibitem{nazer2009compute}
B.~Nazer and M.~Gastpar, ``{Compute-and-forward: Harnessing interference
  through structured codes},'' \emph{Arxiv preprint arXiv:0908.2119}, 2009.

\bibitem{bresler_tse:2008}
\BIBentryALTinterwordspacing
G.~Bresler, A.~Parekh, and D.~Tse, ``The approximate capacity of the
  many-to-one and one-to-many gaussian interference channels,'' 2008. [Online].
  Available: \url{http://arxiv.org/abs/0809.3554}
\BIBentrySTDinterwordspacing

\bibitem{jafar:very_strong_IC}
\BIBentryALTinterwordspacing
S.~Sridharan, A.~Jafarian, S.~Vishwanath, and S.~A. Jafar, ``Capacity of
  symmetric k-user gaussian very strong interference channels.'' [Online].
  Available: \url{http://arxiv.org/abs/0808.2314}
\BIBentrySTDinterwordspacing

\bibitem{Narayanan:2008}
\BIBentryALTinterwordspacing
M.~P. Wilson, K.~Narayanan, H.~Pfister, and A.~Sprintson, ``Joint physical
  layer coding and network coding for bi-directional relaying,'' 2008.
  [Online]. Available: \url{http://arxiv.org/abs/0805.0012}
\BIBentrySTDinterwordspacing

\bibitem{nam:2009bit}
\BIBentryALTinterwordspacing
W.~Nam, S.-Y. Chung, and Y.~Lee, ``Capacity of the {G}aussian two-way relay
  channel to within 1/2 bit,'' 2009. [Online]. Available:
  \url{http://arxiv.org/abs/0902.2438}
\BIBentrySTDinterwordspacing

\bibitem{Kim:sarnoff}
S.~Kim, N.~Devroye, P.~Mitran, and V.~Tarokh, ``Comparison of bi-directional
  relaying protocols,'' in \emph{Proc. IEEE Sarnoff Symposium}, Princeton, NJ,
  Apr. 2008.

\bibitem{gunduz2010multi}
D.~Gunduz, A.~Yener, A.~Goldsmith, and H.~Poor, ``{The multi-way relay
  channel},'' \url{http://arxiv.org/abs/1004.2434/}.

\bibitem{sezgin2010divide}
A.~Sezgin, A.~Avestimehr, M.~Khajehnejad, and B.~Hassibi,
  ``{Divide-and-conquer: Approaching the capacity of the two-pair bidirectional
  Gaussian relay network},'' \emph{Arxiv preprint arXiv:1001.4271}, 2010.

\bibitem{zamir-lattices}
R.~Zamir, ``{Lattices are everywhere},'' in \emph{4th Annual Workshop on
  Information Theory and its Applications, UCSD}, 2009.

\bibitem{nam:2009nested}
\BIBentryALTinterwordspacing
W.~Nam, S.-Y. Chung, and Y.~Lee, ``Nested lattice codes for gaussian relay
  networks with interference,'' 2009. [Online]. Available:
  \url{http://arxiv.org/PS cache/arxiv/pdf/0902/0902.2436v1.pdf}
\BIBentrySTDinterwordspacing

\bibitem{ozgur2010approximately}
A.~Ozgur and S.~Diggavi, ``{Approximately achieving Gaussian relay network
  capacity with lattice codes},''
  \url{http://arxiv4.library.cornell.edu/abs/1005.1284}.

\bibitem{Cover:1979}
{T. M. Cover and A. El Gamal}, ``Capacity theorems for the relay channel,''
  \emph{IEEE Trans. Inf. Theory}, vol.~25, no.~5, pp. 572--584, Sept. 1979.

\bibitem{xie2007network}
L.~Xie, ``{Network coding and random binning for multi-user channels},'' in
  \emph{Proc. of CWIT}, 2007, pp. 85--88.

\bibitem{avestimehr2009capacity}
A.~Avestimehr, A.~Sezgin, and D.~Tse, ``{Capacity of the two-way relay channel
  within a constant gap},'' \emph{European Transactions in Telecommunications},
  2009.

\bibitem{loeliger1997averaging}
H.~Loeliger, ``{Averaging bounds for lattices and linear codes},'' \emph{IEEE
  Trans. Inf. Theory}, vol.~43, no.~6, pp. 1767--1773, 1997.

\bibitem{poltyrev1994coding}
G.~Poltyrev, ``{On coding without restrictions for the AWGN channel},''
  \emph{IEEE Trans. Inf. Theory}, vol.~40, no.~2, pp. 409--417, 1994.

\bibitem{Coppel}
W.~Coppel, \emph{Number Theory: An introduction to mathematics}, 2nd~ed.\hskip
  1em plus 0.5em minus 0.4em\relax Springer, 2009.

\bibitem{Baik:2008}
I.~Baik and S.-Y. Chung, ``Network coding for two-way relay channels using
  lattices,'' in \emph{Proc. IEEE Int. Conf. Commun.}, Beijing, May 2008.

\bibitem{ong2010capacity}
L.~Ong, C.~Kellett, and S.~Johnson, ``{Capacity Theorems for the AWGN Multi-Way
  Relay Channel},'' \url{http://arxiv4.library.cornell.edu/abs/1004.2300}.

\bibitem{Han:1984}
T.~Han, ``A general coding scheme for the two-way channel,'' \emph{IEEE Trans.
  Inf. Theory}, vol. IT-30, pp. 35--44, Jan. 1984.

\bibitem{lim2010noisy}
S.~Lim, Y.~Kim, A.~El~Gamal, and S.~Chung, ``{Noisy network coding},''
  \url{http://arxiv.org/abs/1002.3188}, 2010.

\bibitem{Rankov:2006}
B.~Rankov and A.~Wittneben, ``Achievable rate regions for the two-way relay
  channel,'' in \emph{Proc. IEEE Int. Symp. Inf. Theory}, Seattle, July 2006,
  pp. 1668--1672.

\end{thebibliography}
\end{document}